\documentclass[runningheads]{llncs}
\usepackage{graphicx}
\usepackage[utf8]{inputenc}
\usepackage[T1]{fontenc}
\usepackage[english]{babel}
\usepackage[T1]{fontenc}
\usepackage{lmodern}
\usepackage{colortbl}
\usepackage{lastpage}
\usepackage{amscd}
\usepackage{algpseudocode}
\usepackage{multirow}
\usepackage{amsmath}
\usepackage{amsfonts}
\usepackage{amssymb}
\usepackage{latexsym}
\usepackage{mathrsfs}
\usepackage{float}
\usepackage{multirow}
\usepackage{wrapfig}
\usepackage{color,framed}
\usepackage{graphics}
\usepackage{graphicx}
\usepackage{graphicx}
\usepackage{amsmath}
\usepackage{mathdots}

\DeclareMathOperator{\lcm}{lcm}
\DeclareMathOperator{\diag}{diag}

\begin{document}
	\title{ Monomial codes under linear algebra point of view}
	\author{El Mahdi Mouloua \and
		Mustapha Najmeddine\and
		Maria Isabel Garcia-Planas
		\and Hassan Ouazzou }
	\authorrunning{E. Mouloua et al.}

    \institute{}

	\maketitle              % typeset the header of the contribution
	\begin{abstract}
	
			The monomial codes over a Galois field $\mathbb{F}_q$ that can be thought invariant subspaces are essential to us in this study. More specifically, we look into the link between monomial codes and characteristic subspaces and the decomposition of monomial codes into minimal invariant subspaces. Additionally, we study some of the characteristics of monomial codes and generalize them by proposing the idea of generalized monomial codes.

	\end{abstract}
	%% put here the subject class

	%\2010mathclass{11T71, 	11T71.}
	
	%
	%
	%
\section{Introduction}\label{sec1}
Let $\mathbb{F}_q$ be the finite field with $q$ elements, where $q$ is a prime power, and $n$ be a positive integer such that $\mbox{\rm gcd}\,(q,n)=1$. A code $C$ of length $n$ and dimension $k$ is called linear if $C$ is a linear subspace of $\mathbb{F}_q^n$. The code $C$ is referred to as a $[n,k,d]$ linear code if $d$ is its minimum Hamming distance.

An $[n,k]$ linear code $C$ is called cyclic if whenver $c=(c_{_{0}}, c_{_{1}}, \ldots , c_{_{n-1}})$ belongs to $C$, then :$sc = (c_{_{n-1}}, c_{_{0}}, \ldots , c_{_{n-2}})$ is also in $C$. A cyclic code $C$ of length $n$ can be seen as an ideal of the ring $\frac{\mathbb{F}_q[x]}{<x^{n}-1>}$.
Cyclic codes are simple to implement and rich in algebra. For more details, readers can see \cite[Chapter 4 and 5]{Huffuman2003}.

E.R Berlekamp presented the first generalization of this concept in \cite{Berlekamp} defining constacyclic codes. The work of \cite{Maria2017} introduced monomial codes as generalizations of cyclic and constacyclic codes.

A code $C$ of length $n$ over the field $\mathbb{F}_q$ is called {\it monomial} with respect to the vector $a := (a_{_0},\ldots , a_{_{n-1}}) \in  \mathbb{F}_{q}^n,$if whenever $c = (c_{_{0}}, c_{_{1}}, \ldots , c_{_{n-1}})$ belongs to $C$, then we have $(a_{_{n-1}}c_{_{n-1}}, a_{_{0}}c_{_{0}}, \ldots , a_{_{n-2}}c_{_{n-2}})$ is also in $C$.

Cyclic codes are monomial codes induced by the vector $a=(1,\ldots,1)$. If $a=(1,\ldots,1,\lambda)$ we refer to $\lambda$-constacyclic codes.

As mentioned in \cite{Maria2017}, monomial codes are used widely because shift registers can encode them.

The remainder of this paper is organized as follows: Section 2 provides a brief overview of monomial codes. We define monomial codes as the direct sum of minimal invariant subspaces. Some properties of monomial codes are derived. The goal of Section 3 is to determine the relationship between hyperinvariant and characteristic subspaces in specific situations. The relationship between monomial codes and characteristic subspaces is described. Section 4 presents a generalization of monomial codes with some properties. Finally, section 5 brings the paper to a close.

\section{Monomial codes as invariant subspaces}\label{sec2}
Let $\mathbb{F}_{q}$ be the finite field of $q$ elements where $q$ is a  power prime, let  $ \mathbb{F}_{q}^n$ be the vector subspace of $n$-tuples, assume that $\mbox{\rm gcd}\,(q,n) = 1$ because $x^n-1$ has no repeated irreducible factors over $\mathbb{F}_q$.

Let us start by recalling the definition of monomial matrices and some of their properties.
\begin{definition}[\cite{Garcia2015}]
	
	\begin{enumerate}
		\item  An $n\times n $ matrix  $ A= \left( a_{ij} \right)_{1\leq i,j\leq n}  $   is called to be a \textbf{ monomial }  if it is a  nonsingular matrix and has in each row and each column exactly one non-zero component.
		\item An $n\times n $ matrix  $ P= \left( P_{ij} \right)_{1\leq i,j\leq n}  $ called to be  \textbf{permutation  matrix}  if there is a permutation $\sigma \in S_{n}$ such that $P$ is  	
		obtained  by permuting the columns  of the identity matrix $I_n$ (i.e)
		for each  $P_{ij} =\begin{cases}
		1 \text{ if }\  i = \sigma(j)\\
		0 \text{ if }\  i \neq  \sigma(j)
		
	\end{cases},$
	and we write $P=P_{\sigma}$.
	
\end{enumerate}
\end{definition}

\begin{remark}
Each permutation matrix $P_{\sigma}$ is a monomial matrix with all non-zero components equal to $1$.
\end{remark}
According to \cite{Garcia2015} we  have the following lemma:
\begin{lemma}[\cite{Garcia2015}, Lemma 2]
If  $A$ is a   monomial matrix of order $n$ with non-zero components $a_0, a_1,\ldots, a_{n-1}$ elements of $ \mathbb{F}_{_q},$   then there is  a permutation  $\sigma\in S_n$ such that
$A= \diag(a_{_0}, a_{_1},\ldots, a_{_{n-1}}) P_{\sigma} $
\end{lemma}
Taking into account that the set of monomial matrices has a multiplicative group structure.

We are interested in monomial matrices of the form:
\begin{equation}\label{Monomial_matrix}
A = \left( \begin{array}{ccccc}
0 & 0 & \ldots &0 & a_{_{n-1}} \\
a_{_0} & 0 &\ldots  &0 & 0 \\
0 & a_{_1} & \ddots &\vdots &\vdots \\
\vdots& \ddots &\ddots  &0 &\vdots\\
0& \ldots &0 & a_{_{n-2}} & 0
\end{array}\right)
\end{equation}
We will call it a {\it simple monomial matrix}. Below we summarize some of the important properties of this kind of matrix.

%%%%%%%%%%%%%%%%%%%%%%%%%%%%%% Proposition 1 %%%%%%%%%%%%%%%%%%%%%%%%%%%%%%%%%%%%%%%%%%%%%%%%%	
\begin{proposition}[\cite{Maria2017} properties]
Let $A$ be a simple monomial matrix, with non-zero components $a_0, a_1,\ldots, a_{n-1} \in \mathbb{F}_{_q}$. Then,

\begin{itemize}
	\item[-] $A^{n} = a_{_0}...a_{_{n-1}}I_n$
	\item[-] $A^{-1} = \frac{1}{\prod_{i=0}^{_{n-1}}a_i } A^{_{n-1}}.$ ($\prod_{i=0}^{_{n-1}}a_i \neq 0$ because  $A$ is nonsingular)
	%\item[-] The characteristic polynomial is $p_{\bar(a)}(s)=\det(A-sI_n)=(-1)^{n}(s^n-\prod_{i=1}^{n}a_i).$
	\item[-] Suppose that $a = \prod_{i=0}^{n-1}a_i  $. Then $A$ is similar to $A_{a}$, where :
	\begin{equation*}
		A_{a} =
		\begin{pmatrix}
			0& 0 & \cdots & 0 &  a\\
			1 & 0 & \cdots & 0  & 0\\
			0 & 1 & \cdots & 0  & 0\\
			\vdots  & \vdots  & \ddots & \vdots&0  \\
			0 & 0 & \cdots & 1&0
		\end{pmatrix}
	\end{equation*}
	and
	\begin{equation*}
		S =
		\begin{pmatrix}
			0& 0 & \cdots & 0 &  \prod_{i=0}^{n-1}a_i\\
			a_1 & 0 & \cdots & 0  & 0\\
			0 & a_1a_2 & \cdots & 0  & 0\\
			\vdots  & \vdots  & \ddots & \vdots&0  \\
			0 & 0 & \cdots & \prod_{i=0}^{n-2}a_i&0
		\end{pmatrix}
	\end{equation*}
\end{itemize}

\end{proposition}

In \cite{Maria2017}, the authors define the structure of monomial codes:
\begin{definition}[\cite{Maria2017}, Definition 3.1 page 1103]\label{D4.1}	
A linear code  $C \subseteq {F_q}^{^n}$  is called
\textbf{monomial} code with associated vector  $a=\left(a_{0},a_{1},\ldots, a_{n-1}\right)\in {F_q}^{*^{n}} $
if for each codeword $c=(c_0,c_1,\ldots, c_{n-1})\in C,$ we have
$c^{'}= \left(a_{n-1}c_{n-1},a_{0}c_{0},\ldots,a_{n-2}c_{n-2} \right)$ is also a codeword.\\
The shift (the map $c 	\longrightarrow sc)$ can be represented in a matrix form: $$Ac^t = (a_{_{n-1}}c_{_{n-1}},a_{_0}c_{_0},\ldots,a_{_{n-2}}c_{_{n-2}}).$$  $A$ is given below:
\begin{equation}
A = \left( \begin{array}{ccccc}
0 & 0 & \ldots &0 & a_{_{n-1}} \\
a_{_0} & 0 &\ldots  &0 & 0 \\
0 & a_{_1} & \ddots &\vdots &\vdots \\
\vdots& \ddots &\ddots  &0 &\vdots\\
0& \ldots &0 & a_{_{n-2}} & 0
\end{array}\right)
\end{equation}
\end{definition}
Let:

\begin{equation}\label{}
\begin{array}{ccccc}
\Phi_{a}:& \mathbb{F}_{q}^n   &\longrightarrow &    \mathbb{F}_{q}^n\\\\
& x=(x_0,x_2,\ldots,x_{n-1})   & \longmapsto &
(a_{n-1}x_{n-1},a_{1}x_1,..,a_{n-2}x_{n-2})
\end{array}
\end{equation}
whose associated matrix  in the canonical basis of $\mathbb{F}_{q}^n$, is exactly the  matrix $A$ above. \\

It is clear that $\Phi_{a}$ is an homomorphism of $\mathbb{F}_q^{n}$.

As a direct result, we have the following proposition:

%%%%%%%%%%%%%%%%%%%%%%%%%%%%%% Proposition 2 %%%%%%%%%%%%%%%%%%%%%%%%%%%%%%%%%%%%%%%%%%%%%%%%%
\begin{proposition}
A linear code $C$ with length $n$ over the field $\mathbb{F}_{q}$ is monomial if, and only if, $C$ is an $A$-invariant subspace of $\mathbb{F}^n_{q}$.
\end{proposition}

Let $f_{_A}(x)$ be the characteristic polynomial of the matrix $A$, then  we have:
$$f_A(x)=\begin{vmatrix}
-x & 0 & \ldots &0 & a_{n-1} \\
a_0 & -x &\ldots  &0 & 0 \\
0 & a_1 & \ddots &\vdots &\vdots \\
\vdots& \ddots &\ddots  &-x &\vdots\\
0& \ldots &0 & a_{n-2} & -x
\end{vmatrix} = (-1)^n\left (x^n-\prod_{i=0}^{n-1}a_i\right )$$

In the remainder $f(x)$ means $f_{_A}(x)$.

The eigenvectors associated with the eigenvalues of the matrix $A$ allow us to describe invariant subspaces. Below is a proposition that gives the structure of these eigenvalues:

%%%%%%%%%%%%%%%%%%%%%%%%%%%%%% Proposition 3 %%%%%%%%%%%%%%%%%%%%%%%%%%%%%%%%%%%%%%%%%%%%%%%%%		
\begin{proposition}
For each eigenvalue $\lambda$ (in the case where there exists
any) the vector $$v(\lambda)=(\lambda^{n-1}, a_{0}\lambda^{n-2}, a_{0}a_{1}\lambda^{n-3},\ldots , a_{0}a_{1}\ldots a_{n-2})$$ is an associated eigenvector.
\end{proposition}
\begin{proof}
Taking into account that the matrix $A$ is nonsingular $a_{0}\cdot a_{1}\cdot \ldots \cdot a_{n-1}\cdot \lambda\neq 0$, so $v(\lambda)\neq 0$.

$$\begin{pmatrix}
0 & 0 & \ldots &0 & a_{n-1} \\
a_0 & 0 &\ldots  &0 & 0 \\
0 & a_1 & \ddots &\vdots &\vdots \\
\vdots& \ddots &\ddots  &0 &\vdots\\
0& \ldots &0 & a_{n-2} & 0
\end{pmatrix}\begin{pmatrix}
\lambda^{n-1}\\ a_{0} \lambda^{n-2}\\a_{0}a_{1}\lambda^{n-3}\\ \vdots\\ a_{0}\ldots a_{n-2}
\end{pmatrix} =\begin{pmatrix} a_{n-1}a_{0}\ldots a_{n-2}\\
a_{0} \lambda^{n-1}\\ a_{0}a_{1} \lambda^{n-2}\\a_{0}a_{1}a_{2}\lambda^{n-3}\\ \vdots\\ a_{0}\ldots a_{n-2}\lambda
\end{pmatrix}=\lambda \begin{pmatrix}
\lambda^{n-1}\\ a_{0} \lambda^{n-2}\\a_{0}a_{1}\lambda^{n-3}\\ \vdots\\ a_{0}\ldots a_{n-2}
\end{pmatrix}
$$

(Note that $\prod_{i=0}^{n-1} a_{i}=\lambda^{n}$)
\end{proof}

Let  $f(x) = (-1)^{n}f_1(x)\ldots f_r(x)$ be the factorization of $f$
into irreducible factors over $\mathbb{F}_q$. By the theorem of Cayley Hamilton, we have $f(A) = 0$.
Assume that $(n,q)=1$, thus $f(x)$ has distinct factors $ f_i(x)$, $i=1,\ldots,r$ and consider the homogeneous set of equations:
\begin{equation}\label{space-equation}
f_i(A)x=0, \quad  x \in \mathbb{F}_q^n
\end{equation}

\noindent for $i = 1, \ldots, r$. If $W_i$ stands for the solution space of \eqref{space-equation}, then we may write $W_i = \mbox{\rm Ker}\,  f_{i}(\Phi_a)$.
To investigate the decomposition of monomial codes as minimal invariant subspaces, we need the following proposition (to see \cite{Radkova2009}).
%%%%%%%%%%%%%%%%%%%%%%%%%%%%%% Proposition 4 %%%%%%%%%%%%%%%%%%%%%%%%%%%%%%%%%%%%%%%%%%%%%%%%%
\begin{proposition}\label{prop}
Let $\Phi$ be an homomorphism of $V$ and let $U$ be a $\Phi$-invariant subspace of $V$ and $\dim_{F}V = n$. Then $f_{\Phi \vert U} (x)$ divides
$f_{\Phi(x)}$. In particular, if $V = U \bigoplus W$ and $W$ is a $\Phi$-invariant subspace of $\mathbb{F_q}^n$ then $f_{\Phi(x)} = f_{\Phi \vert U}f_{\Phi \vert W}$. And $Z = Z_1 \bigoplus \ldots \bigoplus Z_r$.
\end{proposition}
Below we summarize the general properties of minimal invariant subspaces; for the proof of Theorem 1, we adopt the same approach proposed in \cite{Radkova2007}:
%%%%%%%%%%%%%%%%%%%%%%%%%%%%%% Theorem 1 %%%%%%%%%%%%%%%%%%%%%%%%%%%%%%%%%%%%%%%%%%%%%%%%%	
\begin{theorem}\label{theo 1}
\begin{enumerate}
	\item For each $i=1,\ldots , r$, $W_i$ is a $\Phi_{a}$-invariant  subspace of $\mathbb{F}_q^n$.
	\item If $Z$ is a $\Phi_{a}$-invariant subspace of $\mathbb{F}_q^n$, and $Z_i = Z \cap W_i$, for $i = 1 \ldots r$ then, $Z_i$ is   $\Phi_a$-invariant subspace of $\mathbb{F}_q^n$ and $Z = Z_1 \bigoplus \ldots  \bigoplus Z_r$.
	\item $\mathbb{F}_q^n = W_1 \bigoplus \ldots \bigoplus W_t$.
	\item  $\dim_{\mathbb{F}_q}W_i = \mbox{\rm deg}\, f_i(x) = k_i$.
	\item $f_{\Phi_a \vert W_i}(x) = (-1)^{k_i}f_i(x)$.
	\item $W_i$ is a minimal $\Phi_a$-invariant subspace of $\mathbb{F}_q^n$.
\end{enumerate}
\end{theorem}
\begin{proof}
\begin{enumerate}
	\item  Assume that $w \in W_i$ then $f_i(A)w=0$,
	and so $f_i(A)\Phi_a(w)=f_i(A)Aw=Af_i(A)w=0 $.
	Thus $\Phi_a(w) \in W_i$ and so $W_i$ is an $\Phi_a$-invariant subspace.
	
	\item Set $\Tilde{f}_i(x)=\frac{f(x}{f_i(x)}$ for i= 1, \ldots , r , since $(\Tilde{f}_1(x),\ldots , \Tilde{f}_r(x)) = 1$, we can found some polynomials $a_1(x),\ldots , a_r(x)$ $\in \mathbb{F}_q[x]$.
	such that: $a_1(x)\Tilde{f}_1(x)+\ldots + a_r(x)\Tilde{f}_r(x)=1$, thus for every $z \in Z$ we have:
	$z= a_1(A)\Tilde{f}_1(A)z+\ldots+ a_r(A)\Tilde{f}_r(A)z$, set $z_i =a_i(A)\Tilde{f}_i(A)z  \in Z$. Then $f_i(A)z=a_i(A)f(A)z=0$, thus $z_i \in W_i \cap Z$, as a result we confirm that $Z = Z_1 + \ldots + Z_r.$
	Suppose that $z \in Z_i \cap \sum_{j \ne i}Z_j$, then $f_i(A)z=0$ and $\Tilde{f}_i(A)z = 0$, since $\mbox{\rm gcd}\,(f_i(x),\Tilde{f}_i(x))=1$, we have polynomials $a(x)$ and $b(x)$, $\in \mathbb{F}_q[x]$ such that $a(x)f_i(x) + b(x)\Tilde{f}_i(x)=1$ thus $z = a(A)f_i(A)z + b(A)\Tilde{f}_i(A)z = 0$. Finally $Z_i \cap \sum_{j \ne i}Z_j = {0}$ and so $Z = Z_1 \bigoplus \ldots  \bigoplus Z_r $.
	\item We just apply the assertion 2 with $Z=\mathbb{F}_q^n$.
	\item Let $m \geqslant 0$ be the smallest integer positive such that $c , \Phi_a(c) , \Phi_a^2(c), \ldots , \Phi_a^{m}(c)$ are linearly dependent, for all $c \in W_i$ then there exist elements $c_0, \ldots , c_{m-1} \in \mathbb{F}_q$ such that:
	\begin{equation}
	\Phi_a^m(c) = c_{k-1}\Phi_a^{m-1}(c) + \ldots + c_{2}\Phi_a^{2}(c)+c_1\Phi_a(c)+c_0c.
	\end{equation}
	Let $t(x)= x^m-c_{m-1}x^{m-1}-....-c_0$ be a polynomial element of $\mathbb{F}_q$. \\
	Since $(t(\Phi_a))(c) = 0$ and $(f_i(\Phi_a))(c) = 0$, it follows that $[(t(x), f_i(x))(\Phi_a)](c) = 0$, or $(t(x), f_i(x))$  is equal to 1 or to $f_i (x)$. thus $(t(x), f_i(x)) = f_i(x)$ and so $f_i(x)$ divides $t(x)$.
	as a result $k_i\leqslant \mbox{\rm deg}\, t(x) = k$, so  the vectors c, $\Phi_a(c), \ldots , \Phi_a^{k_i}(c)$ are
	linearly dependent, since $c \in W_i$ then $(f_i(\Phi_a))(c) = 0$, and from the minimality of $k$ we obtain
	$k = k_i$ . We know that $\dim W_i \geqslant k_i$,  $n = \dim_{\mathbb{F}_q}{\mathbb{F}_q^n}=\sum_{i=1}^{r}\dim_{\mathbb{F}_q}W_i \geqslant \sum_{i=1}^{r}k_i=\mbox{\rm deg}\, f = n$. Then $\dim_{\mathbb{F}}W_i=k_i$.

	\item Assume that $g^{(i)} = (g_{1}^{(i)}, \ldots , g_{k_i}^{(i)})$ is a basis of $W_i$ over $\mathbb{F}_q$, $i = 1, \ldots , r$ and let $A_i$ be the matrix of $\Phi_a \vert W_i$ with respect to that basis. set $\dot{f_i} = f_{\Phi_a \vert W_i}$. Assume that $( \dot{f_i}  , f_i ) = 1$.
	Hence there are polynomials $u(x), v(x) \in \mathbb{F}_q[x]$, such that $u(x)\dot{f_i}  + v(x)f_i(x) = 1$.
	Then $u(A_i)\dot f_i(A_i ) + v(A_i )f_i(A_i ) = I_n$. By the theorem of Cayley-Hamilton $\dot f_i(A_i )=0$, therefore $b(A_i)f_i(A_i ) = E$. Let us verify
	that $f_i(A_i ) = 0$, and we find a contradiction.
	
	By the assertion 3 we obtain that $g = (g_1^{(1)} , \ldots , g_{(k_i)}^{(1)}, \ldots , g_{1}^{(r)},\ldots , g_{k_r}^{(r)})$ is a basis of $\mathbb{F}_q^n$ and we can represent $\Phi_a$ with respect to the basis $g$ by the following matrix:
	\begin{equation}
	A' = \left( \begin{array}{ccccc}
	A_1 &  &  &  &  \\
	& A_2 &  &  &  \\
	&  & \ddots & &  \\
	&  &  & A_{r-1} &  \\
	&  &  &  & A_{r} \\
	\end{array} \right)
	\end{equation}
	And since $A'=P^{-1}AP$ where P is the transformation
	matrix from the standard basis of $\mathbb{F}_q^n$ to the basis $g$, thus:
	$$f_i(A')=\left( \begin{array}{ccccc}
	f_i(A_1) &  &  &  &  \\
	& f_i(A_2) &  &  &  \\
	&  & \ddots & &  \\
	&  &  & f_i(A_{r-1}) &  \\
	&  &  &  & f_i(A_{r}) \\
	\end{array} \right) = f_i(P^{-1}AP)= P^{-1}f_i(A)P $$
	Let $g_j^{(i)} = \lambda_{j_1}^{(i)}e_1 + \ldots + \lambda_{j_n}^{(i)}e_n , j = 1, \ldots , k_i$. Since $g_j \in W_i$ thus:
	$$
	f_i(A')
	\begin{pmatrix}
	0 \\[3mm]
	\vdots \\[3mm]
	1\\[3mm]
	\vdots \\[3mm]
	0
	\end{pmatrix}
	= P^{-1}f_i(A)P\begin{pmatrix}
	0 \\[3mm]
	\vdots \\[3mm]
	1\\[3mm]
	\vdots \\[3mm]
	0
	\end{pmatrix} =  P^{-1}f_i(A)\begin{pmatrix}
	\lambda_{j_1}^{(i)}\\[6mm]
	\vdots \\[6mm]
	\lambda_{j_n}^{(i)}
	\end{pmatrix}  = 0 $$
	
	Where 1 is on the $(k_1 + \ldots + k_{i-1} + j)$-th position. According to the last equation
	$f_i(A_i) = 0$, and $(f_i , \dot{f}_i ) \neq 1$. Since $f_i$ and $\dot{f_i}$ are polynomials of same degree $k_i$ and $f_i$ is monic and irreducible, we obtain that $f_{\Phi_a \vert W_i} = {(-1)}^{k_{i}}f_i$.
	\item Let $W$ be $\Phi_a$-invariant subspace of $\mathbb{F}_{q}^n $ and let ${0} \neq W \subseteq{W_i}$.
	Then by proposition \ref{prop} we have  $f_{\Phi_a \vert W}$ divides $f_i$. Since the polynomial $f_i$ is irreducible, $\dim_{\mathbb{F}_q}W$ = $\dim_{\mathbb{F}_q}W_i$ and $W = W_i$.
\end{enumerate}
\end{proof}

%%%%%%%%%%%%%%%%%%%%%%%%%%%%%%%%%%%  Proposition 5 %%%%%%%%%%%%%%%%%%%%%%%%%%%%%%%%%%%%%
\begin{proposition}\label{prop 1}
Let $U$ be a $\Phi$-invariant subspace of $\mathbb{F}_q^n$. Then $U$ is a direct sum of $\Phi$-invariant subspaces $W_i$ of $\mathbb{F}_q^n$.
\end{proposition}

\begin{proof}
This result is immediate from assertion $2$ of Theorem \ref{theo 1}.
\end{proof}

Now, we investigate the decomposition of monomial codes into minimal invariant subspaces; for the Proof of Theorem 2, we adopt the same approach proposed in \cite{Radkova2009}.
%%%%%%%%%%%%%%%%%%%%%%%%%%%%%% Theorem 2 %%%%%%%%%%%%%%%%%%%%%%%%%%%%%%%%%%%%%%%%%%%%%%%%%	
\begin{theorem}
Let $C$ be a linear monomial code of length n over $\mathbb{F}_q$. then we have the following
facts:
\begin{enumerate}
	\item $C = W_{i_1}  \bigoplus \ldots \bigoplus W_{i_s}$ and $\dim_{\mathbb{F}_q}C = k_{i_1} + \ldots + k_{i_s}.$ Where $k_{i_r}$ is the dimension of $W_{i_r}$
	\item $f_{\Phi_a \vert C} ={(-1)}^{k}f_{i_1}\ldots f_{i_s} = g(x)$
	\item $c \in C$ if and only if $g(A)c = 0$
	\item $f_{\Phi_a \vert C}$ has the smallest degree concerning property (3)
	\item $\mbox{\rm rank}\,(g(A)) =  n-k $
\end{enumerate}
\end{theorem}
%%%%%%%%%%%%%%%%%%%%%%%%%%% Proof of theorem 2 %%%%%%%%%%%%%%%%%%%%%%%%%%%%%%%%%%%%%
\begin{proof}
\begin{enumerate}
	%%%%%%%%%%%%%%%%%%%%%%%%%%%%%% Property 1 %%%%%%%%%%%%%%%%%%%%%%%%%%%%%%%%%%%%%%%%%%%
	\item A direct application of the theorem\ref{theo 1}.
	
	%%%%%%%%%%%%%%%%%%%%%%%%%%%%%% Property 2 %%%%%%%%%%%%%%%%%%%%%%%%%%%%%%%%%%%%%%%%%%%
	\item Let  $(g_1^{(i_r)}, \ldots , g_{k_{i_r}}^{(i_r)})$ be a basis of $W_{i_r}$ over $\mathbb{F}_q$, $r = 1, \ldots , s$ and let $A_{i_r}$ be the matrix of $\Phi_{a} \vert W_{i_r}$ with respect to that basis set $\dot{f_{i_r}} = f_{\Phi_a \vert W_{i_r}}$. Then, $(g_1^{(i_1)}, \ldots , g_{k_{i_1}}^{(i_1)}, \ldots , g_1^{(i_s)}, \ldots , g_{k_{i_s}}^{(i_s)})$ is a basis of $C$ over $\mathbb{F}_q$ and $\Phi_{a} \vert C$ is represented to that basis, by the following matrix:
	\begin{equation}
	A' = \left( \begin{array}{ccccc}
	A_{i_1} &  &  &  &  \\
	& A_{i_2} &  &  &  \\
	&  & \ddots & &  \\
	&  &  & A_{i_{s-1}} &  \\
	&  &  &  & A_{i_s} \\
	\end{array} \right)
	\end{equation}
	$f_{\Phi_{a} \vert C}(x)=\dot{f_{i_1}}(x)\ldots \dot{f_{i_s}}(x)=(-1)^{k_{i_1}+\ldots+k_{i_1}}f_{i_1}\ldots f_{i_s}$.
	
	%%%%%%%%%%%%%%%%%%%%%%%%%%%%%% Property 3 %%%%%%%%%%%%%%%%%%%%%%%%%%%%%%%%%%%%%%%%%%%
	\item Let $a \in C$ then $a = w_{i_1}+\ldots+w_{i_s}$ for where $w_{i_r} \in W_{i_r}$ and $r = 1 ,\ldots s$, and $g(A)a = {(-1)}^{k}[(f_{i_1}\ldots f_{i_s})(A)w_{i_1}+\ldots+(f_{i_1}\ldots f_{i_s})(A)w_{i_s}]=0$.
	Conversely, assume that $g(A)a=0$ for all $a \in C$ for some $a \in \mathbb{F_q}^n$ and let verify that $a \in C$
	by applying the property $3$ of theorem \ref{theo 1} we have $a = w_{i_1}+\ldots+w_{i_t}$, where $w_{i_r} \in W_{i_r}$
	$g(A)a =  {(-1)}^{k}[(f_{i_1}\ldots f_{i_s})(A)w_1 + \ldots + (f_{i_1}\ldots f_{i_s})(A)w_t]$.
	So that $g(A)(w_{j_1} + \ldots + w_{j_l})$ where $\{j_1, \ldots , j_l\} \in \{1, \ldots ,t\}\textbackslash \{i_1, ...,i_s\}.$
	Let $b = w_{j_1} + \ldots + w_{j_l}$ and $h(x) = \frac{{(-1)}^n(x^n-\prod_{i=0}^{n-1}a_i)}{g(x)}$. Since $(h(x), g(x)) = 1$, there are polynomials
	$E(x), F(x) \in \mathbb{F}_q[x]$ such that $E(x)h(x) + F(x)g(x) = 1$. Therefore, $v = E(A)h(A)v +
	F(A)g(A)v = 0$ and so $a \in C$.
	
	%%%%%%%%%%%%%%%%%%%%%%%%%%%%%% Property 4 %%%%%%%%%%%%%%%%%%%%%%%%%%%%%%%%%%%%%%%%%%%
	\item  Suppose that $e(x) \in \mathbb{F}_q[x]$ is a nonzero polynomial of smallest degree such that
	$e(A)a = 0$ for all $a \in C$. By the division algorithm in $\mathbb{F}_q[x]$ there are polynomials
	$q(x), r(x)$ such that $g(x) = e(x)q(x) + r(x)$, where $deg(r(x)) < deg(e(x))$. Then for
	each vector $a \in C$, we have $g(A)a = q(A)e(A)a + r(A)a$ and hence, $r(A)a = 0$.
	However, this contradicts the $e(x)$ choice unless $r(x) = 0$. Thus $e(x)$ divides $g(x)$. Then
	$e(x)$ is a product of some irreducible factors of $g(x)$, and without loss of generality
	we may assume that $e(x) = {(-1)}^{(k_{i_1}+k_{i_2} +\ldots  +k_{i_m})}f_{i_1}f_{i_2}\ldots f_{i_m}
	$ and $m < s$. Let  $\widehat{C}=W_{i_1}\bigoplus \ldots \bigoplus W_{i_m}$, clearly $\widehat{C} \subseteq C$ then $e(x) = f_{\Phi_a \vert \widehat{C}}$.and since $g(A)a = 0$ for all $a \in C$ we obtain that
	$C \subseteq \widehat{C}$. This contradiction proves the statement.
	
	%%%%%%%%%%%%%%%%%%%%%%%%%%%%%% Property 5 %%%%%%%%%%%%%%%%%%%%%%%%%%%%%%%%%%%%%%%%%%%
	\item By property (3), C is the solution space of the homogeneous set of equations $g(A)a = 0$. Then:
	$\dim_{\mathbb{F}_q}C = k = n - \mbox{\rm rank}(g(A))$, which proves the statement.
\end{enumerate}
\end{proof}

%%%%%%%%%%%%%%%%%%%%%%%%%%%%%% Example for theorem 2 %%%%%%%%%%%%%%%%%%%%%%%%%%%%%%%%%%%%%%%%%%%%%%%%%		
\begin{example}
Let consider the particular case where $n =4$, $q=5$ and $a = (a_1=1, a_2=1, a_3=1, a_4=3)$.
Then:

\begin{equation}
A = \begin{pmatrix}
0 & 0 & 0 & 1  \\
3& 0 & 0 & 0  \\
0 & 4 &0 &0     \\
0 & 0 & 3 & 0\\
\end{pmatrix}
\end{equation}

In $F_5[t]$, $P_{A}(t) = (t^4 + 4) = (t+1)(t+2)(t+3)(t+4)$
Let $C$ be a monomial code over $\mathbb{F}_{5}$  defined as follows:
$C = \mbox{\rm ker}\,(A +I) \bigoplus \mbox{\rm ker}\,(A +2I)  = <1,2,2,4> \bigoplus <1,1,3,3>  = \mbox{\rm ker}\,(A^2+3A+2I) =\mbox{\rm ker}\, \begin{pmatrix}
2 & 0 & 3 & 3  \\
4& 2 & 0 & 3  \\
2 & 2 &2 &0     \\
0 & 2 & 4 & 2\\
\end{pmatrix} =<(1,3,2,1),(0,4,1,1)>$ and so $\dim_{\mathbb{F}_q^n}(C) = 2$.

\end{example}
%%%%%%%%%%%%%%%%%%%%%%%%%%%%%%%%%%%%%%%%%%%%%%%%%%%%%%%%%%%%%%%%%%%%%%%%%%%%%%%%%%%%%%%%%%%%%%%%%%%%%%%%%%%%%%%%%

%%%%%%%%%%%%%%%%%%%%%%%%%%%%%%%  section 3 %%%%%%%%%%%%%%%%%%%%%%%%%%%%%%%
\section{Characteristic subspaces and monomial codes}
Let $A$ be a linear map of a finite-dimensional vector
space $\mathbb{F}_{q}^n$ over a field $\mathbb{F}_{q}$. An $A$-invariant subspace is called characteristic if it is invariant under all automorphisms that commute with $A$. If the invariant subspace remains invariant for all linear maps of
$\mathbb{F}_{q}^n$ that commute with $A$ then the subspace is called hyperinvariant for $A$.

\begin{remark}
\begin{itemize}
	\item[-]
	Let $F$ be an $A$-invariant subspace, then $XF$ is $A$-invariant for all $X\in C(A)$.
	\item[-]  $F$ is an $A$ invariant subspace if and only if $\bar F=S^{-1}F$ is $J$-invariant with $A=SJS^{-1}$. Then, we can consider the matrix $A_{a}$.
\end{itemize}
\end{remark}

In order to construct characteristic subspaces, it is helpful to describe  the centralizer of the matrix $A$

Remember that the centralizer $C(A)$ of $A$ is the set of the matrices $X$  commuting with $A$.

\begin{proposition}
\begin{itemize}
	\item[a)] The centralizer $C(A_{a})$ of $A_{a}$ is the set of matrices $X$ with
	\begin{equation*}
		X =
		\begin{pmatrix}
			x_{n-1}& ax_0 & ax_1 & ax_2 & \cdots & ax_{n-3} & ax_{n-2}\\
			x_{n-2}& x_{n-1} & ax_0 & ax_1 & \cdots & ax_{n-4} & ax_{n-3}\\
			
			\vdots  &  & \ddots & \ddots& & \\
			\vdots  &   &  & \ddots&\ddots & \\	
			x_2& x_3 & x_4 & x_5 & \cdots & ax_{0} & ax_{1}  \\
			x_1& x_2 & x_3 & x_4 & \cdots & x_{n-1} & ax_{0} & \\
			x_0& x_1 & x_2 & x_3 & \cdots & x_{n-2} & x_{n-1} & \\
		\end{pmatrix} = x_{n-1}I+x_{n-2}A_{a}+\ldots +x_{0}A^{n-1}_{a}
	\end{equation*}
	\item[b)] Let $\bar A =TA_{a}T^{-1}$  be an equivalent matrix to $A_{a}$. Then $C({\bar A})=\{Y\mid Y=
	TXT^{-1}\}$.
\end{itemize}
\end{proposition}
\begin{proof}
a) It suffices to solve the system $XA_{a}=A_{a}X$.

\noindent		b) $XA_{a}= XT^{-1}AT =T^{-1}ATX$, $TXT^{-1}AT T^{-1}=TT^{-1}ATXT^{-1}$.
\end{proof}

Let $x=(x_{1},\ldots , x_{n})$ be a vector in $\mathbb{F}_{q}^n$ and consider the subspace $C_{A_{0}}(x)$ of $\mathbb{F}_{q}^n$ generated by $x$:  $C_{A_{0}}(x)= \{x,A_{0}x,\ldots , A_{0}^{r}x,\ldots \}$

Cayley-Hamilton theorem ensures that there exists a $r\leq n-1$ such that $C_{A_{0}}(x)=[x,A_{0}x,\ldots ,A_{0}^rx]$, and consider $r$ the least number with this property.

\begin{corollary}
The subspace $C_{A_{0}}(x)$ is characteristic subspace. (In fact, it is an hyperinvariant subspace).
\end{corollary}
\begin{proof}
$$\begin{array}{l} (x_{n-1}I+x_{n-2}A_{a}+\ldots +x_{0}A^{n-1}_{a})(\alpha_{0}x+\alpha_{1}A_{0}x+\ldots +\alpha_{r}A^rx)=\\
x_{n-1}\alpha_{0}x+(x_{n-1}\alpha_{1}+x_{n-2}\alpha_{0})\ldots +x_{0}\alpha_{r}A_{0}^{n+r-1}x=\\ \sum_{i=0}^r\beta_{i}A^ix\in C_{A_{0}}(x).\end{array}$$
\end{proof}
\begin{corollary}
Let $A$ be a simple monomial matrix. Then, the subspace $C_{A}(x)$ is a characteristic subspace.
\end{corollary}
\begin{proof}
It is enough to take into account that $A=TA_0T^{-1}$ and repeat the previous calculations for this case.

\end{proof}
\begin{proposition}
Let $F=\mbox{\rm Ker}\, (A-\lambda I)$ be the subspace of  eigenvectors of $A$ of eigenvalue $\lambda $ of $A$. Then, $F$ is a characteristic subspace of $A$.

In general,  let $G=\mbox{\rm Ker}\, (p(A))$  be a subspace for some polynomial in $A $. Then $G$ is a characteristic subspace.
\end{proposition}
\begin{proof}
$\forall v \in F$, and for all $X\in C(A)$, $ AXv=XAv=X(\lambda v)=\lambda Xv$, then $Xv\in F$.

$\forall v\in G$, and for all $X\in C(A)$, $ p(A)Xv=Xp(A)v=X0=0$, then $Xv\in G$.

\end{proof}

\begin{theorem} A linear code $C$ with length $n$ over the field $\mathbb{F}_{q}$ is monomial if, and only if, $C$ is an $A$-characteristic subspace of $\mathbb{F}^n_{q}$.
\end{theorem}
%%%%%%%%%%%%%%%%%%%%%%%%%%%%%%%  section 3 %%%%%%%%%%%%%%%%%%%%%%%%%%%%%%%

%%%%%%%%%%%%%%%%%%%%%%%%%%%%%%%  section 4 %%%%%%%%%%%%%%%%%%%%%%%%%%%%%%%
\section{Generalized monomial codes}

\begin{definition}\label{GMC}	
A linear code $C$ of length $n$ over the field $\mathbb{F}_q$ is called generalized monomial if, and only if, $C$ is an invariant subspace of $\mathbb{F}_q^n$ for any monomial matrix.
\end{definition}

In order to obtain the invariant subspaces of any monomial matrix, we analyze the characteristic polynomial of the matrix.

It is known that any permutation $\sigma$  of $\{1,\ldots , n\}$  can be written as a product
of disjoint cycles, then we have that
\begin{proposition}
Any  monomial matrix can be written as a product
of disjoint  simple monomial matrices (as \eqref{Monomial_matrix}).
\end{proposition}

\begin{proposition}
The characteristic polynomial of a  monomial matrix is the product of the characteristic polynomials of the simple monomial factors.
\end{proposition}

The factorization of the characteristic polynomial into irreducible
factors depends not only on the factors simple monomial factors but on the finite field.\\
Assume that $\sigma=\sigma_1 \sigma_2\ldots \sigma_r, $ the product of $r$ disjoint of  $n_i-$cycle  $\sigma_i $, and let $\sigma_i = (i_1,...,i_{n_{i}})$ be a cycle of length $n_{i}$ in the decomposition of $\sigma$ into $r-$ disjoint cycles; Being $a_{i_1},a_{i_2},\ldots,a_{i_{n_{i}}}$ the coefficients of the monomial matrix in columns $i_1,...,i_{n_{i}}$.
% To the notation above and according to authors in \cite{Garcia2015}, we have the following theorem about the characteristic polynomial of each cycle in the case where the permutation $\sigma$ is decomposed into disjoint cycles. \\
%\begin{theorem}
%	The characteristic polynomial of each cycle $\sigma_i$ of length $n_i$ is: $$\chi_{P_{\sigma_{i}}}(s)=s^{n_i}-\prod_{k=1}^{n_i}a_{i_k}$$.
%\end{theorem}
%For any monomial matrix $M$   we have :
%\begin{equation}
%\chi_{_{M(s)}} = \chi_{_{P_{_{\sigma_1}}}}(s) \times \ldots \times \chi_{_{P_{_{\sigma_r}}}}(s).
%\%end{equation}
%According to authors in \cite{}, If $\sigma$ has a cycle structure of $(C_1, C_2, \ldots , C_n)$, then the characteristic
%polynomial of $P_{\sigma}$ is :
%\begin{equation}
%\chi_{_{P_{\sigma}}}(s)={(-1)}^n\prod_{k=1}^n{(s^k-1)}^{C_k}.
%\end{equation}
In order to find new subspaces invariant under monomial matrices, we present  the definition  of a minimal polynomial of a permutation,
\begin{definition}[minimal polynomial of $\sigma$]
The minimal polynomial $m_{\sigma}(x)$ of a given permutation $\sigma $ is defined to be the minimal polynomial of its associated matrix $P_{\sigma}.$
\end{definition}

\begin{proposition}\label{PP.1}
Let $\sigma$ be a permutation of length $n,$ with minimal polynomial $ m_{\sigma}(x)$.
\begin{enumerate}
	\item  $  m_{\sigma}(x)=x^n-1 $ if $\sigma$ is a $n-$cycle.
	\item  $  m_{\sigma}(x)=x^{\lcm(n_1,n_2,\ldots,n_r)}-1 ,$ if $\sigma=\sigma_1\sigma_2\ldots \sigma_r$  a product of  $n_i-$cycle $\sigma_i.$
\end{enumerate}
\end{proposition}

\begin{proof}
\begin{enumerate}
	\item Assume that $\sigma $ of the  trivial form (i.e) : \\ $\sigma(1,2,\ldots,n)= (n,1,\ldots,n-1), $ then  $P_{\sigma} $  is the companion matrix of $x^n-1,$ and so $ m_{\sigma}=x^n-1.$. On the other hand if
	$\sigma$ hasn't the trivial form then there is a permutations $ \sigma^{'},\ \rho $ such that $ \sigma=  \sigma^{'-1}\circ \rho \circ  \sigma $,  where $\rho$ is the trivial $n-$cycle. Hence
	$ P_{\sigma}= P_{\sigma^{'-1}} P_{\rho}  P_{\sigma^{'}}= P_{\sigma^{'}}^{-1} P_{\rho}  P_{\sigma^{'}} ,$ and so
	$ m_{\sigma}(x)=m_{\rho}(x)=x^n-1$.
	
	\item   We can write   $\sigma=\sigma_1\sigma_2\ldots \sigma_r$ as a product of $r$ $n_i-$cycle $\sigma_i.$ Then   $m_{\sigma_i}(x)=x^{n_i}-1, $ and so
	$ m_{\sigma}(x)= \lcm_{1\leq i\leq r}\left\lbrace x^{n_i}-1\right\rbrace = x^{\lcm(n_1,n_2,\ldots,n_r)}-1$.
\end{enumerate}
\end{proof}

Below we give a decomposition of monomial codes under the case where $\sigma$ the product of $r$ disjoint of  $n_i$-cycle  $\sigma_i $:
\begin{theorem}\label{TT.2}	
Let   $ C\subseteq \mathbb{F}_{_q}^{^n} $  be a monomial code and   $\sigma=\sigma_1 \sigma_2\ldots \sigma_r, $ the product of $r$ disjoint of  $n_i$-cycle  $\sigma_i $, then $C$ can be decomposed as
$$ C= C_1\oplus C_2\oplus \ldots \oplus C_r$$
where each $C_i$ is monomial.
\end{theorem}
\begin{proof} Let for each $i=1,\ldots ,r$ be the minimal polynomial $\sigma_i.$ So,  as $  \sigma_i, \ 1\leq i\leq r $ are disjoint then on can write decompose  $\mathbb{F}_{_q}^{^n}$ as
$$ \mathbb{F}_{_q}^{^n}=  \ker(m_1(\sigma))\oplus \ker(m_2(\sigma)) \oplus \ldots \oplus \ker(m_r(\sigma)).$$
It follows that $$ C= C\cap \mathbb{F}_{_q}^{^n}= C\cap \ker(m_1(\sigma))\oplus  C\cap\ker(m_2(\sigma)) \oplus \ldots \oplus C\cap \ker(m_r(\sigma)). $$
Clearly  $C_i= C\cap \ker(m_r(\sigma))$ are invariant by $ \sigma_{/_{\ker(m_i(\sigma))}}$ hence it is invariant by $\sigma_i,$ and so $C_i$ is invariant by $\sigma_i$.
\end{proof}

The property that an invariant subspace under a simple monomial matrix is a characteristic (in fact, hyperinvariant subspace) that cannot be generalized to a generalized monomial

\begin{example}
$\mathbb{F}[5]$, $n=6$, $\mbox{\rm gcd}\,(5,6)=1$

$A= \begin{pmatrix}
0  &   0  &   1  &   0  &   0 &    0\\
2  &   0  &   0  &   0  &   0   &  0\\
0  &   3  &   0  &   0  &   0   &  0\\
0  &   0  &   0  &   0  &   0   &  1\\
0  &   0  &   0  &   2  &   0   &  0\\
0  &   0  &   0  &   0  &   3   &  0\end{pmatrix}$

$V=[(1,2,1,0,0,0)]$ is an $A$-invariant subspace

Let

$ C=\begin{pmatrix}
0  &  0  &   0   &  1  &   3   &  1\\
0  &  0 &   0  &   2  &   1   &  2\\
0   &  0 &   0   &  1   &  3    & 1\\
1   &  3  &   1  &   0  &   0  &  0\\
2   &  1  &   2  &   0  &   0   &  0\\
1   & 3   &  1   &  0  &  0  & 0
\end{pmatrix}\in C(A)$

$  \begin{pmatrix}
0  &  0  &   0   &  1  &   0  &  0\\
0  &  0 &   0  &   0  &   1   &  0\\
0   &  0 &   0   &  0   &  0   & 1\\
1   &  0  &   0  &   0  &   0  &  0\\
0  &  1  &   0  &   0  &   0   &  0\\
0  & 0 &  1   &  0  &  0  & 0
\end{pmatrix}\begin{pmatrix}1\\2\\1\\0\\0\\0 \end{pmatrix}= \begin{pmatrix}
0\\0\\0\\ 1\\ 2\\1
\end{pmatrix}\notin V$.

\noindent	So, the subspace is not characteristic.
\end{example}

In a more general way,  computing the centralizer of a generalized monomial code:

$\begin{pmatrix}
A_{1}&&\\ &\ddots &\\ &&A_{r}
\end{pmatrix}$

with $A_{i}$ simple monomial matrices

$$\begin{pmatrix}
A_{1}&&\\ &\ddots &\\ &&A_{r}
\end{pmatrix}\begin{pmatrix}
X_{11}&\ldots &X_{1r}\\ \vdots &&\vdots\\ X_{r1}&\ldots &X_{rr}
\end{pmatrix}- \begin{pmatrix}
X_{11}&\ldots &X_{1r}\\ \vdots &&\vdots\\ X_{r1}&\ldots &X_{rr}
\end{pmatrix}\begin{pmatrix}
A_{1}&&\\ &\ddots &\\ &&A_{r}
\end{pmatrix}=0.$$

Equivalently,

$$\begin{array}{rl}
A_{i}X_{ii}-X_{ii}A_{i}&=0, i=1,\ldots , r\\
A_{i}X_{ij}-X_{ij}A_{j}&=0, i, j=1,\ldots , r, i\neq j
\end{array}$$

$A_{i}X_{ii}-X_{ii}A_{i}$ corresponds to the centralizer of simple monomial matrices

suppose now $A_{1},\ldots , A_{r}\in M_{n}(\mathbb{F})$, and  $gcd\, (f_{A_{i}}(x),f_{A_{j}}(x)=1$ for all $i,j=1,\ldots , r, i\neq j$

\begin{proposition}
In these conditions the equation $A_{i}X_{ij}-X_{ij}A_{j}=0, i, j=1,\ldots , r, i\neq j$ has a unique solution $X_{ij}=0$
\end{proposition}

\begin{proof}
$f_{A_{i}}(x)=x^n-\alpha_{i}$ with $\alpha_{i}\neq \alpha_{j}$ for all $i\neq j$.

From $A_{i}X=XA_{j}$ premultiplying  by $A_{i}$:

$A_{i}^2X=A_{i}XA_{j}$ changing $A_{i}X$ by $XA_{j}$ we have $A_{i}^2X= XA_{j}^2$.

Repeating the process, we have

$$A^n_{i}X=XA^n_{j}$$

Cayley-Hamilton theorem  ensures $A_{i}^n=\alpha_{i}I$ and $A_{j}^n=\alpha_{j}I$

Then, $\alpha_{i}IX=X\alpha_{j}I$
but,
equivalently: $(\alpha_{i}-\alpha_{j})X=0$  $(\alpha_{i}-\alpha_{j})\neq 0$, So $X=0$
\end{proof}

\section{Conclusion}
The class of monomial codes was presented, and some new results of such codes were found. The concept of the characteristic subspace is discussed and linked to monomial codes, which are considered invariant subspaces under special homomorphism. Finally, the new concept of generalized monomial code is presented with some new properties.

%\section{Bibliography}

%\begin{thebibliography}{1}

%\end{thebibliography}

\end{document}